\documentclass[a4paper, 10pt]{amsart}

\usepackage{amsmath, amssymb, latexsym, amsthm}
\usepackage[usenames, dvipsnames]{color}
\usepackage{lmodern}
\usepackage{enumitem}
\usepackage{mathtools}
\mathtoolsset{mathic=true}

\DeclareSymbolFont{usualmathcal}{OMS}{cmsy}{m}{n}
\DeclareSymbolFontAlphabet{\mathcal}{usualmathcal}

\allowdisplaybreaks
\DeclareMathOperator{\minimize}{minimize}
\DeclareMathOperator{\sbjto}{subject\,to}
\DeclareMathOperator{\trace}{trace}
\DeclareMathOperator{\sat}{sat}
\DeclareMathOperator{\rank}{rank}

\newcommand{\N}{\mathbb{N}}
\newcommand{\Nz}{\N_0}
\newcommand{\R}{\mathbb{R}}

\newcommand{\EE}{\mathsf{E}}
\newcommand{\PP}{\mathsf{P}}

\newcommand{\ol}{\overline}

\newcommand{\opt}{^\star}
\newcommand{\transp}{^\top}
\newcommand{\inverse}{^{-1}}
\newcommand{\psinv}{^+}
\newcommand{\Let}{\coloneqq}
\newcommand{\teL}{\eqqcolon}
\newcommand{\eps}{\varepsilon}

\newcommand{\rhpolicy}{\hat\pi}
\newcommand{\reachindex}{\kappa}
\newcommand{\reachab}{\mathcal{R}}
\newcommand{\setmin}{\setminus}
\newcommand{\drv}{\mathrm{d}}

\renewcommand{\ge}{\geqslant}
\renewcommand{\le}{\leqslant}
\renewcommand{\mapsto}{\longmapsto}

\newcommand{\ConSet}{\mathbb{U}}
\newcommand{\NoiseSet}{\mathbb{W}}
\newcommand{\lra}{\longrightarrow}
\newcommand{\join}{\sharp}

\newcommand{\abs}[1]{\left\lvert{#1}\right\rvert}
\newcommand{\norm}[1]{\left\lVert{#1}\right\rVert}

\newcommand{\epower}[1]{\mathrm{e}^{#1}}

\newcommand{\indic}[1]{\mathtt{1}_{#1}}
\newcommand{\secref}[1]{\S\ref{#1}}
\newcommand{\pmat}[1]{\begin{pmatrix}#1\end{pmatrix}}

\numberwithin{equation}{section}

\newtheorem{theorem}{Theorem}[]

\newtheorem{proposition}[theorem]{Proposition}

\theoremstyle{remark}
\newtheorem{Example}[theorem]{Example}
\newtheorem{remark}[theorem]{Remark}

\newtheoremstyle{dcdef}{8pt}{4pt}{}{}{\scshape}{.}{ }{\thmname{#1}\thmnote{ (\mdseries #3)}}
\theoremstyle{dcdef}

\newtheoremstyle{nonum}{8pt}{4pt}{}{}{\scshape}{.}{ }{\thmname{#1}\thmnote{ (\mdseries #3)}}
\theoremstyle{nonum}
\newtheorem{assumptionnn}{Assumption}

\newcommand{\AssumptionEnd}{\hspace{\stretch{1}}{$\diamondsuit$}}
\newcommand{\ExampleEnd}{\hspace{\stretch{1}}{$\triangle$}}
\newcommand{\RemarkEnd}{\hspace{\stretch{1}}{$\triangleleft$}}

\usepackage[colorlinks=true,
		raiselinks=true,
		linkcolor=MidnightBlue,
		citecolor=Mahogany,
		urlcolor=ForestGreen,
		pdfauthor={Debasish Chatterjee},
		pdftitle={},
		pdfkeywords={},
		pdfsubject={Technical Report},
		plainpages=false]{hyperref}

\title[Stability and performance of stochastic predictive control]{Stability and performance of stochastic predictive control}

\author[D.\ Chatterjee]{Debasish Chatterjee}
\address{Systems \& Control Engg.\\ 104 SysCon\\ IIT Bombay, Powai\\ Mumbai 400076\\ India}
\urladdr{\url{http://www.sc.iitb.ac.in/~chatterjee}}
\email{dchatter@iitb.ac.in}

\author[J.\ Lygeros]{John Lygeros}
\address{Automatic Control Lab.\\ ETL I22, ETH Z\"urich\\ Physikstrasse 3\\ 8092 Z\"urich\\ Switzerland}
\urladdr{\url{http://control.ee.ethz.ch/~lygeros}}
\email{jlygeros@control.ee.ethz.ch}

\thanks{This work was carried out during D.\ Chatterjee's visit to the Automatic Control Laboratory, ETH Z\"urich, in June-August 2012, supported by the Swiss National Science Foundation via the Short International Visit Fellowship no.\ $\text{IZK0Z2}\_142780$.}

\date{19 April 2013}

\begin{document}

	\begin{abstract}
		This article is concerned with stability and performance of controlled stochastic processes under receding horizon policies. We carry out a systematic study of methods to guarantee stability under receding horizon policies via appropriate selections of cost functions in the underlying finite-horizon optimal control problem. We also obtain quantitative bounds on the performance of the system under receding horizon policies as measured by the long-run expected average cost. The results are illustrated with the help of several simple examples.
	\end{abstract}
	\maketitle

	\section{Introduction}
		With the steady growth in the availability of fast computing machines, control techniques that involve algorithmic selection of actions that minimize some performance objective have gained prominence. Receding horizon predictive control, which is based on such algorithmic selection procedures, has evolved over the years into one of the most useful and applicable control synthesis techniques currently available to a control engineer; see e.g., \cite{ref:MayRaw-09} for a survey of the modern theory and applications in the deterministic setting. Stochastic versions of receding horizon techniques initially evolved within the operations research community, see e.g., \cite{ref:hernandez-lerma1, ref:hernandez-lerma2}, with inventory and manufacturing systems as primary application areas, and have steadily filtered into the domain of control systems, with current applications in financial engineering, process control, industrial electronics, power systems, etc.

		While the deterministic and robust versions of receding horizon control techniques have become standardized and are well-documented, the available literature on the stochastic version still lacks a comprehensive and systematic treatment. Especially prominent in this regard is the matter of stability of control systems under stochastic receding horizon control; indeed, most of the literature does not appear to take advantage of the significantly developed and advanced results on stability of Markov processes. Chief among the reasons for this discrepancy between the deterministic and stochastic settings, perhaps, is the fact that the technical nature of the arguments involved in the stochastic version of stability is significantly heavier than its deterministic counterpart. Indeed, while the bare-essential arguments involved in establishing Lyapunov stability of discrete-time deterministic dynamical systems are only a few and are quite classical, the technical arguments and conditions in the theory of stability of Markov chains is by far larger in number, and constitute an active area of research even today. In addition to that, one has a diverse library of notions of stability that are peculiar to the stochastic setting, and are simply non-existent in the deterministic or the robust setting.

		This article is an attempt at bridging this gap---we connect receding horizon control techniques to some of the principal elements of the theory of stability of Markov processes. Motivated by, and in the spirit of \cite[\S3]{ref:Mayne2000}, first we systematically develop a framework for studying stability of discrete-time controlled stochastic systems under receding horizon policies. We critically examine two approaches in this connection, namely, ensuring stability by appropriate selection of the cost functions, and by adjoining an appropriate constraint to the underlying finite-horizon optimal control problem, before focussing on the former. Against the backdrop of certain standard (and no-so-standard) conditions for stability of Markov processes, we establish conditions on the cost functions such that these stability conditions are satisfied. Thus, this selection procedure, by design, ensures that the closed-loop system under the corresponding receding horizon control policy is stable. We utilize theorems on stability of Markov processes off-the-shelf as to this end. As such, the results pertaining to stability presented here should be regarded as representative guidelines---rather than offer a set of stand-alone results, we provide a general framework for establishing stability results. The details for specific applications must be worked out on a case-by-case basis, as we illustrate through several examples.

		In addition, we develop a framework for analyzing the performance of the closed-loop systems under stochastic receding horizon control policies. Selecting a long-run expected average cost derived from the underlying finite-horizon optimal control problem as our performance index, we provide quantitative bounds on this performance index under receding horizon policies and mild hypotheses. Observe that receding horizon policies are extracted from a finite-horizon optimal control problem, and as such do not naturally offer any clue concerning the long-run expected average costs that they incur. The relationship between stability and performance is also explored here. In particular, we obtain a bound on the aforementioned performance index under a receding horizon policy that also ensures stability in an appropriate sense.

		The layout of this article is as follows: \secref{s:def} provides the description of the control systems. Our results on stability under receding horizon control are contained in \secref{s:stab}, while performance bounds are provided in \secref{s:perf}. Several examples illustrate our results throughout \secref{s:stab} and \secref{s:perf}. The proofs of our results are provided in the Appendices \secref{s:app:stab} and \secref{s:app:perf}. The emphasis here is on conceptual clarity and a systematic presentation sans heuristics. The setting, insofar as the system, the associated receding horizon problem, and the results are concerned, is at an abstract level; this choice is targeted at conveying the key ideas in a transparently clear fashion, without the overload of excessive notation. In particular, the ideas presented here can be readily generalized to the setting of Markov decision processes; we choose to stay with simpler notation and technical requirements here. Numerical tractability of the underlying optimal control problems, which is an integral aspect of receding horizon control techniques, is \emph{not} addressed here.

	\section{System description}
	\label{s:def}

		Consider the discrete-time dynamical system given by the recursion
		\begin{equation}
		\label{e:gensys}
			x_{t+1} = f(x_t, u_t, w_t),\qquad x_0 \text{ given}, \quad t\in\Nz,
		\end{equation}
		where 
		\begin{itemize}[label=\(\circ\), leftmargin=*]
			\item \(x_t\in\R^d\), \(u_t\in\ConSet\subset\R^m\), and \(w_t\in\NoiseSet\subset\R^p\) are the states, the control actions, and the noise at time \(t\);
			\item \(f:\R^d\times \ConSet \times \NoiseSet\lra\R^d\) is a measurable function;\footnote{Henceforth ``measurability'' on Euclidean spaces will refer to ``Borel measurability.''}
			\item \(\ConSet\) is the (nonempty) \emph{control set}, assumed to be measurable and containing the element \(0\in\R^m\);
			\item \(\NoiseSet\subset\R^p\) is assumed to be a measurable set;
			\item \((w_t)_{t\in\Nz}\) is the \emph{process noise}---a \(\NoiseSet\)-valued random process with the \(w_t\)'s independent and identically distributed.\footnote{All random vectors are assumed to be defined on some underlying probability space, for which \(\PP(\cdot)\) is the probability measure, and \(\EE[\cdot]\) is the expectation under \(\PP\). The assumption that \(w_t\)'s are independent and identically distributed can be substantially weakened at the expense of some notational clutter; we choose to stay with the simpler setting here.}
		\end{itemize}

		Let \(k\) be a positive integer. Recall that a \emph{\(k\)-stage (feedback) policy} is a collection \( \pi_{0:k-1} \Let (\pi_0, \pi_1, \ldots, \pi_{k-1}) \) of measurable functions \(\pi_i : \R^d\lra\ConSet\) for each \(i\); we set the \(t\)-th control action as \(u_t = \pi_t(x_t)\). For the synthesis of control actions in a receding horizon fashion, we consider given:
		\begin{itemize}[label=\(\circ\), leftmargin=*]
			\item a \emph{horizon} \(N\in\N\),
			\item a \emph{cost-per-stage function} \(c:\R^d\times \ConSet\lra [0, +\infty[\) and a \emph{final cost function} \(c_F:\R^d\lra[0, +\infty[\), both assumed to be measurable, and
			\item a class \(\Pi\) of feedback policies.
		\end{itemize}
		We introduce the \(N\)-horizon value function
		\begin{equation}
		\label{e:MPCcost}
			V_N(x, \pi) \Let \EE^\pi_x\biggl[\sum_{i=0}^{N-1} c(x_i, u_i) + c_F(x_N)\biggr]
		\end{equation}
		where the policies \(\pi\) belong to the class \(\Pi\).\footnote{If \(\varphi:\R^d\lra\R\) is a measurable function, then \(\EE^\pi_x[\varphi(x_t)]\) stands for the conditional expectation of \(\varphi(x_t)\) given \(x_0 = x\), where \(x_t\) is the state at time \(t\) under the policy \(\pi\); \(\PP_x\) is the corresponding conditional probability.}

		\begin{assumptionnn}
			Without detailing the specifics,
			\vspace*{-1ex}
			\begin{enumerate}[label={\rm (A\arabic*)}, leftmargin=*, align=left, widest=4]
				\item \label{a:regularity of process} we assume sufficient regularity of the process \((w_t)_{t\in\Nz}\) such that the cost \eqref{e:MPCcost} is finite for all \(x\in\R^d\) and all \(\pi\in\Pi\).\AssumptionEnd
			\end{enumerate}
		\end{assumptionnn}

		With these ingredients, the centerpiece of receding horizon control can be stated: It consists of the \(N\)-horizon optimal control problem:
		\begin{equation}
		\label{e:MPC problem}
		\begin{aligned}
			\minimize	& \quad V_N(x, \pi)\\
			\sbjto		& \quad 
				\begin{cases}
					\pi\in\Pi,\\
					\text{dynamics \eqref{e:gensys}}.
				\end{cases}
		\end{aligned}
		\end{equation}

		\begin{assumptionnn}
			In addition to \ref{a:regularity of process}, we assume that
			\vspace*{-1ex}
			\begin{enumerate}[label={\rm (A\arabic*)}, leftmargin=*, align=left, widest=4, start=2]
				\item \label{a:opt existence} the minimization problem \eqref{e:MPC problem} is well-defined for all \(x\in\R^d\), i.e., for each boundary value \(x\in\R^d\), there exists a policy \(\pi\opt\in\Pi\) that solves \eqref{e:MPC problem}.\AssumptionEnd
			\end{enumerate}
		\end{assumptionnn}
		Conditions under which \ref{a:opt existence} holds are of a technical nature, and well documented, e.g., in \cite[Chapter 3]{ref:hernandez-lerma1}.

		We denote the optimal value function \(V_N(x, \pi\opt)\) by \(V_N\opt(x)\) for all \(x\in\R^d\). The system \eqref{e:gensys} under the optimal policy generates the optimal state trajectory \((x_t\opt)_{t=0}^N\) given by
		\begin{equation}
		\label{e:optimal traj}
			x_{t+1}\opt	= f(x_t\opt, \pi_t\opt(x_t\opt), w_t),\qquad x_0\opt \text{ given}, \quad t = 0, 1, \ldots, N-1.
		\end{equation}
		The technique of receding horizon control consists of applying the first element \(\pi_0\opt\) obtained from the minimization problem \eqref{e:MPC problem} recursively, thereby generating the \emph{receding horizon policy}
		\begin{equation}
		\label{e:rh policy}
			\rhpolicy \Let (\pi_0\opt, \pi_0\opt, \ldots).
		\end{equation}
		To wit, given the state \(x_t\) at time \(t\), one solves the minimization problem \eqref{e:MPC problem} with \(x = x_t\), obtains the optimal policy \(\pi\opt\), applies the first element \(\pi_0\opt\) of the policy, moves to time \(t+1\), and repeats the preceding steps. The system \eqref{e:gensys} under the policy \(\rhpolicy\) generates the state trajectory \((x_t)_{t\in\Nz}\) via the recursion
		\begin{equation}
		\label{e:rh closed loop}
			x_{t+1}	= f(x_t, \pi_0\opt(x_t), w_t),\qquad x_0 \text{ given},\quad t\in\Nz.
		\end{equation}
		Observe that the process \((x_t)_{t\in\Nz}\) generated by \eqref{e:rh closed loop} is \emph{Markovian}, i.e., the probability distribution of the future state \(x_{t+1}\) at time \(t+1\) is conditionally independent of the past \((x_s)_{s=0}^{t-1}\) given the present state \(x_t\). Indeed, for \(S\) a Borel subset of \(\R^d\), we have
		\begin{equation}
		\label{e:Markovian proof}
		\begin{aligned}
			\PP\bigl(x_{t+1} \in S\,\big|\, (x_s)_{s=0}^{t}\bigr) & = \PP\bigl( f(x_t, \pi_0\opt(x_t), w_t) \in S\,\big|\, (x_s)_{s=0}^t\bigr)\\
				& = \PP\bigl( f(x_t, \pi_0\opt(x_t), w_t)\in S\,\big|\, x_t\bigr).
		\end{aligned}
		\end{equation}
		The following sections will study aspects of both qualitative and quantitative behavior of the process \((x_t)_{t\in\Nz}\) generated by \eqref{e:rh closed loop}.

	\section{Stability under receding horizon control}
	\label{s:stab}

		Stability of the \emph{controlled process} \((x_t)_{t\in\Nz}\) generated by the recursion \eqref{e:rh closed loop} is a desirable property in practice. There are two techniques in which stability can be ensured:
		\begin{enumerate}[label=(S\arabic*), align=left, widest=2, leftmargin=*]
			\item \label{stab:cost} \textsf{By appropriate choice of cost functions:} Stability of the controlled process \((x_t)_{t\in\Nz}\) can be ensured by an appropriate selection of the cost-per-stage function \(c\) and the cost function \(c_F\). In the deterministic setting, conditions for asymptotic stability in terms of the cost functions are standard, see e.g., \cite[\S3]{ref:Mayne2000} and \cite{ref:KeeGil-88} for details and further references. In the stochastic setting 
			the aim is to arrive at a Lyapunov-like inequality in terms of the cost functions, which in turns ensures stability of the closed-loop system. While conceptually this technique leads to an elegant analysis, there are two points worthy of note:
				\begin{itemize}[label=\(\circ\), leftmargin=*]
					\item For a control engineer, the selection of the cost functions \(c\) and \(c_F\) is typically dictated by the physics of the problem. In case the stability conditions are not satisfied by the natural candidates \(c\) and \(c_F\), the engineer may be forced to select cost functions that may have little to do with the particular physical aspects of the plant.
					\item The applicability of predictive control is contingent upon numerical tractability of the finite horizon optimal control problem \eqref{e:MPC problem}. The extent of flexibility in the choice of the functions \(c\) and \(c_F\) is determined, therefore, by cases where numerically tractable problems can be derived from \eqref{e:MPC problem}. In other words, applicability of this technique is limited by numerical tractability of the problem \eqref{e:MPC problem}. However, as we shall illustrate through examples, more than one cost functions that ensure stability; there is, therefore, some freedom which the control engineer can utilize to suit numerical tractability.
				\end{itemize}
				In the deterministic setting, standard stability conditions under receding horizon control require the existence of a stabilizing feedback controller \emph{inside} a certain terminal set, such that the terminal set becomes invariant for the controlled system \cite[\S3]{ref:Mayne2000}. In the stochastic setting, in general, even the weakest forms of stability (for instance, positive recurrence, existence of invariant measures, etc,) require the existence of a certain ``drift condition'' \emph{outside} a bounded set \cite{ref:MeyTwe-09}. This difference between the two settings is generally \emph{unavoidable} for the lack of a notion analogous to invariance in the deterministic setting; see also Remark \ref{r:drift condition} for a more specific discussion. For instance, in systems where there is non-zero probability of jumps in the state infinitely often and the magnitude of the jumps is not bounded, the notion of deterministic invariance does not make sense for any bounded subset of the state-space. Of course, this assertion does not apply to systems subjected to bounded noise where it may be possible to perform a robust analysis, but it does indeed apply to the standard benchmark case of a linear control system with additive and independent Gaussian noise.

			\item \label{stab:constraint} \textsf{By adjoining an appropriate constraint to the optimal control problem \eqref{e:MPC problem}:} This technique was first adopted in \cite{ref:BerBem-09, ref:HokChaRamChaLyg-10, ref:HokCinChaRamLyg-12} in the context of receding horizon control of linear stochastic controlled systems. It consists of adjoining a constraint to the optimal control problem \eqref{e:MPC problem}, so that the modified optimal control problem stays feasible for all \(x\in\R^d\), and the resulting receding horizon policy \(\rhpolicy\) defined in \eqref{e:rh policy} ensures stability. Observe that (i) the problem \eqref{e:MPC problem} where one intends to adjoin the constraint is limited to a finite-horizon, while (ii) the target of the constraint---attaining stability of the closed-loop system---involves a necessarily infinite-horizon notion. Two points to note:
				\begin{itemize}[label=\(\circ\), leftmargin=*]
					\item It is imperative to ensure that the problem \eqref{e:MPC problem} with the new constraint is feasible for all boundary values \(x\); this necessarily imposes restrictions on the type of admissible constraints.
					\item Adjoining a constraint to the problem \eqref{e:MPC problem} potentially shifts both the optimal value and the optimizer \(\pi\opt\) corresponding to the original problem \eqref{e:MPC problem}. Therefore, a trade-off between the performance and a certain desirable qualitative behavior of the closed-loop system may have to be accepted.
				\end{itemize}
				As will be evident from the above discussion, a systematic development of the case \ref{stab:constraint} is largely impossible due to the absence of a set of unifying objects inherent to the optimal control problem \eqref{e:MPC problem}. Since the constraints do not, generally, depend on the cost functions, the details of the technique may differ significantly between specific applications
		\end{enumerate}

		In this article we focus on \ref{stab:cost}; the relevant results are presented in \secref{s:stab:cost}. Preparatory to that, in \secref{s:stab:review} we briefly recall certain basic aspects of the general theory of stability of discrete-time Markov processes.

		\subsection{Review of the general theory of stability of discrete-time Markov processes}
		\label{s:stab:review}
		The type of stability that we shall focus on here concerns boundedness of sequences of the form \(\bigl(\EE_x[h(x_t)]\bigr)_{t\in\Nz}\), for appropriate functions \(h:\R^d\lra[0, +\infty[\). For instance, consider \(h(z) = \norm{z}^p\) for \(p > 1\). In view of the fact that\footnote{This identity is an immediate consequence of Fubini's theorem.
			}
		\[
			\EE_x\bigl[\norm{x_t}^p\bigr] = p \int_0^{+\infty} r^{p-1} \PP_x(\norm{x_t} > r)\,\drv r,
		\]
		boundedness of \(\bigl( \EE_x[\norm{x_t}^p]\bigr)_{t\in\Nz} \) implies that the conditional probability, given the initial condition \(x_0 = x\), of the states being at a distance \(r\) from the origin decays faster than \(r^{-p}\) as \(r\) grows large, uniformly over time \(t\). In other words, we have an assertion corresponding to the behavior of the tail of the conditional probability distributions \(\PP_x(\norm{x_t} > r), t\in\Nz\), uniformly over time \(t\). The case of \(p = 2\) is especially prevelant in the literature, and goes under the name of mean-square boundedness. To understand the qualitative behavior of \(\R^d\)-valued Markov processes \((x_t)_{t\in\Nz}\), the general strategy consists of studying the behavior of sequences such as \(\bigl(\EE_x[h(x_t)]\bigr)_{t\in\Nz}\) for ``norm-like'' functions \(h\), and drawing appropriate inferences concerning the former.

		Recall that the process \((x_t)_{t\in\Nz}\) generated by \eqref{e:rh closed loop} is Markovian in view of \eqref{e:Markovian proof}. For discrete-time Markov processes, the theory of stability is extremely well-developed, (see e.g., \cite{ref:MeyTwe-09} for a book-length treatment,) and most of the standard conditions for stability involve what is known as a ``negative drift condition.''\footnote{General results dealing with stability but not relying on negative drift conditions are rare; one example may be found in \cite{ref:ChaPal-10}.} A generic negative drift condition takes the following form:
		\begin{enumerate}[label=(D), leftmargin=*]
			\item \label{cond:drift} there exist measurable functions \(\Xi:\R^d\lra[0, +\infty[\) and \(\Upsilon:\R^d\lra[0, +\infty[\), and a bounded and measurable set \(K\subset\R^d\), such that 
			\[
				\EE_x[\Xi(x_1)] - \Xi(x) \le -\Upsilon(x)\quad\text{for all \(x\not\in K\)}.
			\]
		\end{enumerate}
		Depending on the properties of the functions \(\Xi\) and \(\Upsilon\), it may be possible to assert the type stability of the process \((x_t)_{t\in\Nz}\). Observe that the condition \ref{cond:drift} closely resembles Lyapunov stability conditions for deterministic discrete-time systems.

		Perhaps the most well-known drift condition is contained in the following:
		\begin{proposition}
		\label{p:FL classical}
			Let \((x_t)_{t\in\Nz}\) be a Markov process. Suppose that there exist \(\beta > 0\) and \(\lambda_\circ \in [0, 1[\), a measurable function \(V:\R^d\lra[0, +\infty[\), and a compact set \(K\subset\R^d\) such that \(\EE_x[V(x_1)] \le \lambda_\circ V(x)\) for all \(x\not\in K\), and \(\sup_{x\in K} \EE_x[V(x_1)] = \beta\). Then \(\EE_x[V(x_t)] \le \lambda_\circ^t V(x) + \beta(1-\lambda_\circ)\inverse\) for all \(x\in\R^d\) and \(t\in\Nz\).
		\end{proposition}

		The hypotheses of Proposition \ref{p:FL classical} imply
		\begin{equation}
		\label{e:geometric drift}
			\EE_x[V(x_1)] - V(x) \le -(1-\lambda_\circ) V(x) \quad \text{for all \(x\not\in K\)},
		\end{equation}
		which is sometimes known as a ``geometric drift condition.'' The condition \eqref{e:geometric drift} is strong---the expected value of the function \(V\) decreases by a fraction \(\lambda_\circ\) in one step \emph{for all} boundary conditions \(x\) outside a compact set. See e.g., \cite{ref:MeyTwe-09} for further details, discussions, and applications of Proposition \ref{p:FL classical}.

		Among the weakest drift conditions, we have the following:
		\begin{proposition}
		\label{p:PR conditions}
			Let \((x_t)_{t\in\Nz}\) be a Markov process. Suppose that there exist \(\beta, M, \eps > 0\), a measurable function \(V:\R^d\lra[0, +\infty[\), and a compact set \(K\subset\R^d\) such that
			\begin{gather}
				\label{e:PR:drift} \EE_x[V(x_1)] - V(x) \le -\beta \quad\text{for all \(x\not\in K\)},\quad\text{and}\\
				\label{e:PR:jump} \EE\bigl[\abs{V(x_{t+1}) - V(x_t)}^{2+\eps}\,\big|\, (V(x_s))_{s=0}^{t}\bigr] \le M\quad \text{for all \(t\in\Nz\)}.
			\end{gather}
			Then for each \(x\in\R^d\) the sequence \(\bigl(\EE_x[V(x_t)]\bigr)_{t\in\Nz}\) is bounded.
		\end{proposition}

		Proposition \ref{p:PR conditions} stipulates a \emph{constant} negative drift \eqref{e:PR:drift} outside a compact set, as opposed to a \emph{geometric} negative drift in \eqref{e:geometric drift}. The condition \eqref{e:PR:drift} is rather weak, and the price for weakening the drift condition is the introduction of a uniform bound \eqref{e:PR:jump} on the jumps of the process \((x_t)_{t\in\Nz}\). In general, both the conditions \eqref{e:PR:drift} and \eqref{e:PR:jump} are necessary, and the \((2+\eps)\) exponent in \eqref{e:PR:jump} is tight; see \cite{ref:PemRos-99} for details and (counter-)examples. An application of Proposition \ref{p:PR conditions} to control of linear systems may be found in \cite{ref:RamChaMilHokLyg-10}, and to receding horizon control in \cite{ref:HokChaRamChaLyg-10}.

		Propositions \ref{p:FL classical} and \ref{p:PR conditions} may be viewed as representing two extremes of the spectrum of stability results involving negative drift conditions. We refer the reader to \cite{ref:MeyTwe-09} for other drift conditions and their corresponding assertions concerning stability. Proposition \ref{p:PR conditions} also highlights some of the features peculiar to stochastic control---indeed, while in the deterministic setting, the drift (in terms of Lyapunov functions) needs to be merely negative definite to ensure global asymptotic convergence of the system, the stability assertions in the stochastic setting depend crucially on the functional nature of the drift in addition to other conditions.

		\subsection{Stability under appropriate selection of cost functions}
		\label{s:stab:cost}
		\begin{assumptionnn}
			In addition to \ref{a:opt existence}, we stipulate that
			\vspace*{-1ex}
			\begin{enumerate}[label={\rm (A\arabic*)}, leftmargin=*, align=left, widest=4, start=3]
				\item \label{a:feedback} there exist a measurable feedback control function \(g:\R^d\lra\ConSet\), a number \(b \ge 0\), and a measurable and bounded set \(K\subset\R^d\) such that
					\begin{enumerate}[label=(\ref{a:feedback}-\roman*), leftmargin=*, align=left, widest=2]
						\item \label{e:boundedness condition} \(\displaystyle{\sup_{z\in K}\Bigl\{ c(z, g(z)) - c_F(z) + \EE\bigl[ c_F\circ f(z, g(z), w_0) \bigr]\Bigr\} \le b}\),
						\item \label{e:stabilizing condition} \(c(z, g(z)) + \EE\bigl[c_F\circ f(z, g(z), w_0)\bigr] \le c_F(z)\) for all \(z\not\in K\).\AssumptionEnd
					\end{enumerate}
			\end{enumerate}
		\end{assumptionnn}
		Observe that \ref{e:stabilizing condition} is a negative drift condition in disguise: indeed, it is precisely
		\[
			\EE_x[c_F(x_1^g)] - c_F(x) \le -c(x, g(x))\quad \text{for all \(x\not\in K\)},
		\]
		where \(x_1^g \Let f(x, g(x), w_0)\); here the cost functions \(c_F\) and \(c\) play the roles of the functions \(\Xi\) and \(\Upsilon\) in \ref{cond:drift}, respectively. (``Global'' conditions, similar in spirit to \ref{e:stabilizing condition}, in the context of stochastic receding horizon control have been proposed recently in \cite{ref:QueNes-12}.) However, since the stabilizing feedback \(g\) is not necessarily identical to \(\pi_0\opt\), the condition \ref{e:stabilizing condition} does not guarantee stability under the receding horizon control policy \(\rhpolicy\). Nevertheless, from the boundedness condition \ref{e:boundedness condition} and the drift condition \ref{e:stabilizing condition}, both expressed in terms of the cost functions \(c\) and \(c_F\), we can establish the following drift condition involving the optimal value function \(V_N\opt\) corresponding to \eqref{e:MPC problem}, and under the receding horizon policy \(\rhpolicy\):
		\begin{theorem}
		\label{t:drift}
			Consider the controlled system \eqref{e:gensys} with its accompanying data, and the optimal control problem \eqref{e:MPC problem}. Suppose that Assumption \ref{a:feedback} holds. Then
			\begin{equation}
			\label{e:drift condition}
				\text{for any \(x\in\R^d\),}\quad \EE_x\bigl[V_N\opt(x_1\opt)\bigr] - V_N\opt(x) \le -c(x, \pi_0\opt(x)) + b,
			\end{equation}
			where \((x_t\opt)_{t=0}^{N-1}\) is the sequence generated by the recursion \eqref{e:optimal traj}. In particular, under the receding horizon policy \(\rhpolicy\) derived from \eqref{e:MPC problem}, the closed-loop process \((x_t)_{t\in\Nz}\) generated by \eqref{e:rh closed loop} satisfies
			\begin{equation}
			\label{e:drift condition rh}
				\text{for any \(x\in\R^d\),} \quad \EE^{\rhpolicy}_x\bigl[V_N\opt(x_1)\bigr] - V_N\opt(x) \le -c(x, \pi_0\opt(x)) + b.
			\end{equation}
		\end{theorem}

		Theorem \ref{t:drift} is the first of our two main results, and it is a representative statement aimed at establishing a connection between receding horizon control and stability of Markov processes. Observe that even though \eqref{e:drift condition rh} does not resemble a negative drift condition per se, a condition analogous to \ref{cond:drift} can be extracted from \eqref{e:drift condition rh} under appropriate assumptions on the function \(c\).\footnote{Of course, \eqref{e:drift condition rh} may not lead to a negative drift outside a bounded set, e.g., if the function \(c\) is bounded above by \(b\), and even if a negative drift condition can be extracted from \eqref{e:drift condition rh}, it may not be possible to assert boundedness of the sequence \(\bigl(\EE^{\rhpolicy}_x\bigl[V_N\opt(x_t)\bigr]\bigr)_{t\in\Nz}\), e.g., if the conditions \eqref{e:PR:drift} and \eqref{e:PR:jump} do not hold simultaneously.} Once such a procedure has been carried out, one can apply appropriate results on stability of Markov processes, e.g., Propositions \ref{p:FL classical} or \ref{p:PR conditions}, to assert boundedness of the sequence \(\bigl(\EE^{\rhpolicy}_x\bigl[V_N\opt(x_t)\bigr]\bigr)_{t\in\Nz}\). We shall illustrate applications of Theorem \ref{t:drift} through our results and examples in the sequel.

		In general, analytical expressions of the optimal value functions \(V_N\opt\) are difficult to obtain; however, \(V_N\opt\) can be bounded above and below in terms of the final cost \(c_F\) and the cost-per-stage function \(c\), respectively, as follows:
		\begin{proposition}
		\label{p:V vs final-cost}
			Consider the controlled system \eqref{e:gensys} with its accompanying data, and the optimal control problem \eqref{e:MPC problem}. Suppose that Assumption \ref{a:feedback} holds. Then
			\begin{equation}
			\label{e:V vs final cost}
				\text{for all \(x\in\R^d\),} \quad 
				\begin{aligned}
					& c(x, \pi_0\opt(x)) + (N-1) \inf_{(z, u)\in\R^d\times\ConSet} c(z, u) \\
					& \quad + \inf_{(z, u)\in\R^d\times\ConSet} \EE\bigl[c_F(f(z, u, w_0))\bigr] \le V_N\opt(x) \le c_F(x) + N b.
				\end{aligned}
			\end{equation}
		\end{proposition}

		Proposition \ref{p:V vs final-cost} can be employed in conjunction with Theorem \ref{t:drift} to arrive at stability conditions under additional hypotheses:

		\begin{proposition}
		\label{p:geometric drift}
			Consider the controlled system \eqref{e:gensys} with its accompanying data, and the optimal control problem \eqref{e:MPC problem}. Suppose that Assumption \ref{a:feedback} holds. Assume further that the cost functions satisfy:
			\begin{itemize}[label=\(\circ\), leftmargin=*]
				\item \(\displaystyle{\lim_{\norm{z}\to+\infty} c_F(z) = +\infty}\),
				\item there exist measurable functions \(c_s:\R^d\lra[0, +\infty[\) and \(c_c:\ConSet \lra[0, +\infty[\) such that
					\[
						c(z, v) = c_s(z) + c_c(v)\quad \text{for all }(z, v)\in\R^d\times\ConSet,
					\]
					and
					\[
						\lim_{\norm{z}\to+\infty} c_s(z) = +\infty,
					\]
				\item there exist a constant \(\alpha \in[0, 1[\) and a compact set \(K\subset\R^d\) such that
					\[
						c_s(z) \ge \alpha c_F(z) \quad \text{for all }z\not\in K.
					\]
			\end{itemize}
			Then under the receding horizon policy \(\rhpolicy\), the function \(V_N\opt\) satisfies a geometric drift condition outside some compact subset of \(\R^d\). In particular, for each \(x\in\R^d\) the sequence \(\bigl(\EE^{\rhpolicy}_x\bigl[V_N\opt(x_t)\bigr]\bigr)_{t\in\Nz}\) is bounded.
		\end{proposition}

		\begin{Example}[The LQ problem]
		\label{ex:FL classical}
			Consider the controlled system \eqref{e:gensys} with \(f(x, u, w) = A x + B u + w\) for matrices \(A\in\R^{d\times d}\) and \(B\in\R^{d\times m}\). Let \(\ConSet = \R^m\) and \(\NoiseSet = \R^d\). Let \(w_0\) have a continuous density on \(\R^d\), \(\EE[w_0] = 0\), \(\EE[w_0 w_0\transp] = \Sigma\) for some non-negative definite and symmetric matrix \(\Sigma\in\R^{d\times d}\). Assume that the pair \((A, B)\) is stabilizable. Then, by \cite[Proposition 11.10.5]{ref:Ber-09}, for every symmetric and positive definite matrix \(Q\in\R^{d\times d}\) there exists a matrix \(K\in\R^{m\times d}\) and a symmetric and positive definite matrix \(P\in\R^{d\times d}\) such that \((A + BK)\transp P (A + BK) - P = -Q\). Consider the policy \(\pi = (g, g, \ldots)\), where \(\R^d\ni x\mapsto g(x) \Let Kx\in\R^m\), and define the function \(\R^d\ni x\mapsto V(x) \Let x\transp P x\in[0, +\infty[\). Then it follows from Proposition \ref{p:FL classical} and standard arguments that the closed-loop process \((x_t)_{t\in\Nz}\) under the policy \(\pi\) is stable in the sense that \(\EE_x[V(x_t)] \le \lambda_\circ^t V(x) + \beta(1-\lambda_\circ)\inverse\) for all \(t\in\Nz\),
			\begin{gather*}
				\lambda_\circ = \frac{1}{2}\biggl(1 - \frac{\sigma_{\min}(Q)}{\sigma_{\max}(P)}\biggr),\\
				K = \bigl\{z\in\R^d\,\big|\, x\transp P x \le \tfrac{2}{\lambda_\circ}\trace(P\Sigma)\bigr\},\\
				\beta = \sup_{z\in K}\bigl\{ x\transp(A + BK)\transp P (A + BK) x + \trace(P\Sigma)\bigr\},
			\end{gather*}
			where \(\sigma_{\min}(M)\) and \(\sigma_{\max}(M)\) denote the minimal and maximal singular values of a matrix \(M\).

			In view of the above computations, we define a symmetric and non-negative definite matrix \(R\in\R^{m\times m}\) such that \(K\transp R K \le Q\), where the relation ``\(\le\)'' between the preceding matrices denotes the standard matrix partial order among symmetric non-negative definite matrices. Let us define cost functions \(c(z, u) \Let (1-\alpha) z\transp Q z + \alpha u\transp R u\) and \(c_F(z) \Let z\transp P z\) for \((z, u)\in\R^d\times \R^m\) and \(\alpha\in[0, 1]\). Straightforward calculations show that \ref{e:boundedness condition} and \ref{e:stabilizing condition} hold with \(g(z) = Kz\) and the preceding definitions of \(c\), \(c_F\), \(\beta\), and the compact set \(K\). Consider now the optimal control problem \eqref{e:MPC problem} for a given \(N\in\N\) and the control set \(\ConSet = \R^m\). By Theorem \ref{t:drift}, we see that a receding horizon controller derived from this optimal control problem ensures that \eqref{e:drift condition rh} holds. It is also possible to verify the hypotheses of Proposition \ref{p:geometric drift} in this case, which implies that for each \(x\in\R^d\) the sequence \(\bigl(\EE^{\rhpolicy}_x\bigl[V_N\opt(x_t)\bigr]\bigr)_{t\in\Nz}\) is bounded.
			\ExampleEnd
		\end{Example}

		While Example \ref{ex:FL classical} is entirely standard, it highlights a few noteworthy features of control of linear systems with affine noise, summarized in the following:
		\begin{remark}
		\label{r:drift condition}
		\mbox{}
		\vspace*{-2ex}
		\begin{enumerate}[label=\alph*), leftmargin=*, align=right, widest=b]
			\item In the context of linear systems, the condition \eqref{e:geometric drift} implies that at all states \(x\) of large enough norm, the control action must be strong enough to achieve this geometric decrease. In the absence of a bound on the magnitude of the control actions, it is possible to synthesize linear feedback policies, such that a geometric drift in terms of quadratic functions \(V\) is attained; for instance, consider the feedback policy \((g, g, \ldots)\) with \(g(x) = Kx\) in Example \ref{ex:FL classical}.
			\item \label{r:drift condition:linear} If the control actions are \emph{bounded}, a control policy whose elements are linear maps of the states is inadmissible. In this case, if the noise has unbounded support, e.g., \(w_t\) is a zero-mean Gaussian with a given variance matrix, then the following four cases appear naturally:
				\begin{itemize}[label=\(\triangleright\), leftmargin=*]
					\item if the system matrix \(A\) has an eigenvalue outside the closed unit disc, with no control policy is it possible to ensure a geometric drift condition with quadratic Lyapunov functions \(V\);
					\item if all eigenvalues of \(A\) are inside the open unit disc, then irrespective of the feedback policy, a geometric drift condition can be found for a quadratic Lyapunov function \(V\), as illustrated in Example \ref{ex:FL classical};
					\item if \(A\) is Lyapunov stable, a constant (as opposed to geometric) negative drift condition for the Lyapunov function \(V(z) \Let \norm{z}\) was demonstrated in \cite{ref:RamChaMilHokLyg-10}; we provide a geometric drift condition under the same setting in Proposition \ref{p:ortho stability} below;
					\item if \(A\) is has eigenvalues on the unit circle but with unequal algebraic and geometric multiplicities, the problem of stabilization under bounded controls remains an open problem; see \cite{ref:ChaRamHokLyg-12} for details.
				\end{itemize}
			\item Consider the scalar version of Example \ref{ex:FL classical} with \((w_t)_{t\in\Nz}\) a sequence of mutually independent standard normal random variables. Since 
				\[
					\PP\Bigl(\inf_{t\in\Nz} w_t = -\infty\text{ and }\sup_{t\in\Nz} w_t = +\infty\Bigr) = 1,
				\]
				it is impossible to assert almost sure convergence of the states to any compact set under any policy. For the same reason, it is also impossible to assert a statement of the form
				\[
					\PP\bigl(\text{\(\exists\ t_0\in\N, K\subset\R\) compact,  such that \(x_t \in K\) for all \(t\ge t_0\)} \bigr) \ge 1-\eps
				\]
				for \(\eps\in\:]0, 1[\) preassigned. In other words, under any feedback policy, almost surely, there will be excursions of the states beyond any given compact set infinitely often over an infinite time horizon. An identical assertion carries over to the multidimensional case.\RemarkEnd
		\end{enumerate}
		\end{remark}

		Preparatory to providing further examples illustrating Theorem \ref{t:drift}, we establish a stability result for a particular class of linear systems, as promised in Remark \ref{r:drift condition}-\ref{r:drift condition:linear}: Consider the controlled system
		\begin{equation}
		\label{e:linsys}
			x_{t+1} = A x_t + B u_t + w_t,\qquad x_0 \text{ given}, \quad t\in\Nz,
		\end{equation}
		for given matrices \(A\in\R^{d\times d}\), \(B\in\R^{d\times m}\), and suppose that \((w_t)_{t\in\Nz}\) is a sequence of independent and identically distributed (i.i.d.) random vectors. Suppose that pair \((A, B)\) is controllable with reachability index \(\reachindex\).\footnote{That is, \(\rank\pmat{B & AB & \cdots & A^{\reachindex -1}B} = d\).} We define \(\reachab(A, M) \Let \pmat{A^{\reachindex -1}M & \cdots & AM & M}\) for a matrix \(M\) with \(d\) rows. For a matrix \(M\) we let \(M\psinv\) denote its Moore-Penrose pseudoinverse. Let \(U_{\max} > 0\) be given, and suppose that \(\norm{u_t} \le U_{\max}\) for all \(t\in\Nz\). For \(r > 0\), we define the radial saturation function \(\R^d\ni z\mapsto \sat_r(z) \Let \min\{r, \norm{z}\}\frac{z}{\norm{z}}\) if \(z\neq 0\) and \(0\) otherwise. Let \(I_d\) denote the \(d\times d\) identity matrix. The following proposition proposes a bounded control function such that a geometric drift condition is attained:
		\begin{proposition}
		\label{p:ortho stability}
			Consider the linear controlled system \eqref{e:linsys} with its accompanying data, suppose that \(A\) is orthogonal, and that \(w_0\) is Gaussian with mean \(0\in\R^d\) and given variance \(\Sigma\in\R^{d\times d}\). Let \(\rho \Let \ln\EE\Bigl[\exp\Bigl(\norm{\reachab(A, I_d) \pmat{w_0\transp & \cdots & w_{\reachindex -1}\transp}\transp}\Bigr)\Bigr]\). Suppose that \(U_{\max} > \rho\), define \(V(x) \Let \epower{\norm{x}}\) for \(x\in\R^d\), and let \(K\Let \{z\in\R^d\mid \norm{z} \le 2\rho\}\). Then under the control actions
			\begin{equation}
			\label{e:ortho stability:controls}
				\pmat{u_{\reachindex t}\\ \vdots \\ u_{\reachindex(t+1)-1}} \Let -\reachab(A, B)\psinv \sat_{U_{\max}}\bigl(A^{\reachindex} x_{\reachindex t}\bigr),\quad t\in\Nz,
			\end{equation}
			the closed-loop process \((x_{\reachindex t})_{t\in\Nz}\) is Markovian, and there exists \(\lambda_\circ \in\:]0, 1[\) such that
			\[
				\EE[V(x_{\reachindex (t+1)})\mid x_{\reachindex t}] \le \lambda_\circ V(x_{\reachindex t})\quad\text{on the set }\{x_{\reachindex t}\not\in K\}.
			\]
			In particular, for each \(x\in\R^d\) the sequence \(\bigl(\EE_x\bigl[\epower{\norm{x_{t}}}\bigr]\bigr)_{t\in\Nz}\) is bounded, and the conditionally probability distributions \(\PP_x(\norm{x_t} > r), t\in\Nz,\) have exponentially thin tails uniformly over \(t\in\Nz\).
		\end{proposition}

		Proposition \ref{p:ortho stability} is of independent interest. Stability of \eqref{e:linsys} under bounded controls was considered in \cite{ref:RamChaMilHokLyg-10}, where the authors demonstrated that the same control actions as in \eqref{e:ortho stability:controls}, but under weaker assumptions on the noise,\footnote{To be precise, it was assumed that \(\sup_{t\in\Nz}\EE[\norm{w_t}^4] < +\infty\). The result in \cite{ref:RamChaMilHokLyg-10}, therefore, applies to noise sequences \((w_t)_{t\in\Nz}\) more general than Gaussians.} led to a constant negative drift of the function \(\norm{\cdot}\) outside a certain compact set of the closed-loop sub-sampled process \((x_{\reachindex t})_{t\in\Nz}\). The technical tools in \cite{ref:RamChaMilHokLyg-10} relied on the considerably involved results of \cite{ref:PemRos-99}; in contrast, the proof of Proposition \ref{p:ortho stability} that we provide here relies only on the basic Proposition \ref{p:FL classical}---namely, a geometric drift condition expressed in terms of \(\epower{\norm{\cdot}}\). Note that, Proposition \ref{p:ortho stability} asserts boundedness of \(\bigl(\EE_x\bigl[\epower{\norm{x_t}}\bigr]\bigr)_{t\in\Nz}\), which is a stronger statement compared to, and indeed implies, boundedness of \(\bigl(\EE_x\bigl[\norm{x_t}^2\bigr]\bigr)_{t\in\Nz}\) asserted in \cite{ref:RamChaMilHokLyg-10}. Compare, in particular, that the conditional distributions \(\PP_x(\norm{x_t} > r)\) have exponentially thin tails, by Proposition \ref{p:ortho stability}, if the maximum magnitude of the control actions is large enough, while the main result of \cite{ref:RamChaMilHokLyg-10} asserts that the corresponding distributions have tails that decay faster than inverse quadratically as \(r\) grows large.

		\begin{Example}
		\label{ex:linear scalar}
			Consider the controlled system \eqref{e:gensys} with \(f(x, u, w) = x + u + w\), where \(x, u, w\in\R\). Let \(w_0\) have a continuous density on \(\R\), \(\EE[w_0] = 0\), \(\EE[\abs{w_0}^4] < +\infty\). Fix \(N\in\N\). Let us investigate the possibility of stability under a receding horizon policy derived from the \(N\)-horizon optimal control problem\footnote{Here \(\indic{A}(\cdot)\) denotes the indicator function of the set \(A\), defined as \(\indic{A}(z) = 1\) if \(z\in A\) and \(0\) otherwise.}
			\begin{equation}
			\label{e:scalar opt problem}
			\begin{aligned}
				\minimize	& \quad \EE^{\pi}_x\biggl[\sum_{t=0}^{N-1} \indic{\R\setmin[-2, 2]}(x_t) + \norm{x_N}\biggr]\\
				\sbjto		& \quad \begin{cases}
								\pi\in\Pi			& \text{where \(\Pi\) a class of Markovian policies},\\
								\pi_i(z) \in[-1, 1] & \text{for all \(i = 0, \ldots, N-1\) and \(z\in\R\)},\\
								\text{dynamics \eqref{e:gensys}} & \text{with \(f(x, u, w) = x + u + w\)}.
							\end{cases}
			\end{aligned}
			\end{equation}
			The cost-per-stage function \(c(z, u) = \indic{\R\setmin[-2, 2]}(z)\) ensures that for each realization of the random noise, the cost grows proprtionately to the duration that the state stays out of the set \([-2, 2]\). The policy that solves the minimization problem \eqref{e:scalar opt problem} drives the state inside \([-2, 2]\) as fast as possible, and the final cost \(c_F\) regulates the final state close to \(0\). In the light of Proposition \ref{p:PR conditions} (e.g., following the arguments in \cite{ref:RamChaMilHokLyg-10},) it is not difficult to verify that the policy \((g, g, \ldots)\), with \(g(x) = -\sat(x)\) ensures that\footnote{Here \(\sat(\cdot)\) is the standard saturation function defined as \(\sat(z) = z\) for \(z\in[-1, 1]\), \(1\) if \(z > 1\), and \(-1\) otherwise.}
			\[
				\EE\bigl[\abs{x_{t+1}}\,\big|\, x_{t}\bigr] - \abs{x_{t}} \le -1\quad \text{on the set }\{\abs{x_{t}} > 2\}, \quad t\in\Nz.
			\]
			Note that the closed-loop system is Markovian. One sees after standard computations that the optimal control problem \eqref{e:scalar opt problem} with \(K \Let \{z\in\R\mid \abs{z} \le 2\}\) verifies both \ref{e:boundedness condition} and \ref{e:stabilizing condition}. The selection of the cost-per-stage function \(c\) is not unique; e.g., \(c(z, u) = \frac{1}{2}\bigl(\indic{\R\setmin[-2, 2]}(z) + \indic{[-1, 1]\setmin[-\frac{1}{2}, \frac{1}{2}]}(u)\bigr)\) works just as fine insofar as the matter of satisfying \ref{e:boundedness condition} and \ref{e:stabilizing condition} is concerned. Numerical tractability of the problem \eqref{e:scalar opt problem}, even under the assumption that \((w_t)_{t\in\Nz}\) is a sequence of independent and identically distributed standard normal random variables, is a non-trivial matter.\ExampleEnd
		\end{Example}

		\begin{Example}[Example \ref{ex:linear scalar} continued]
		\label{ex:linear scalar re}
			Consider the controlled system \eqref{e:gensys} with \(f(x, u, w) = x + u + w\), where \(x, u, w\in\R\). Let \((w_t)_{t\in\Nz}\) be a sequence of i.i.d.\ Gaussian random variables with zero-mean and variance \(\sigma^2\) for some given \(\sigma > 0\), and define \(\rho \Let \ln\bigl(\sigma\sqrt{\frac{2}{\pi}}\bigr)\). Let \(U_{\max} > \rho\) be given, and define \(\lambda_\circ \Let \epower{\rho - U_{\max}}\). Fix \(N\in\N\). Let us investigate the possibility of stability under a receding horizon policy derived from the \(N\)-horizon optimal control problem
			\begin{equation}
			\label{e:scalar opt problem re}
			\begin{aligned}
				\minimize	& \quad \EE^{\pi}_x\biggl[(1-\lambda_\circ)\sum_{t=0}^{N-1} \epower{\abs{x_t}} + \epower{\abs{x_N}}\biggr]\\
				\sbjto		& \quad \begin{cases}
								\pi\in\Pi						& \text{where \(\Pi\) a class of Markovian policies},\\
								\abs{\pi_i(z)} \le U_{\max}		& \text{for all \(i = 0, \ldots, N-1\), \(z\in\R\)},\\
								\text{dynamics \eqref{e:gensys}}& \text{with \(f(x, u, w) = x + u + w\)}.
							\end{cases}
			\end{aligned}
			\end{equation}
			Namely, we have the final cost \(c_F(z) = \epower{\abs{z}}\) and the cost-per-stage function \(c(z, u) = (1-\lambda_\circ)\epower{\abs{z}}\). By Proposition \ref{p:ortho stability} it follows that \ref{e:stabilizing condition} holds for some compact \(K\subset\R\), and it is straightforward to verify \ref{e:boundedness condition}. Theorem \ref{t:drift} guarantees that a receding horizon controller derived from \eqref{e:scalar opt problem re} ensures \eqref{e:drift condition rh}. An extension to the multidimensional case can be readily obtained with the technical support of Proposition \ref{p:ortho stability}. It also follows from standard computations that the hypotheses of Proposition \ref{p:geometric drift} hold in this case, which implies that for each \(x\in\R\) the sequence \(\bigl(\EE^{\rhpolicy}_x\bigl[V_N\opt(x_t)\bigr]\bigr)_{t\in\Nz}\) is bounded.
			\ExampleEnd
		\end{Example}


%

	\section{Performance under receding horizon control}
	\label{s:perf}
		In this section we study performance of closed-loop systems under receding horizon control. Our objective here is to arrive at quantitative bounds on the performance of receding horizon policies over an infinite temporal horizon.

		To this end, we must first select a performance index. We contend that the cost-per-stage function \(c\) is a natural candidate with which performance at each time step may be measured. For one, while the final cost \(c_F\) plays an important role in the problem \eqref{e:MPC problem}, the first element \(\pi_0\opt\) of the optimal control policy \(\pi\opt\) enters the function \(c\) but not \(c_F\); since a receding horizon policy is constructed out of \(\pi_0\opt\), the function \(c\) is perhaps a more natural candidate compared to \(c_F\) for measuring performance at each time step. Moreover, the function \(c\) involves both the states and the control actions, while \(c_F\) involves only the states; as such, the expected value of \(c(x_t, u_t)\) at time \(t\) reflects the performance measured with respect to both the states and the control actions.

		In the setting of the dynamical system \eqref{e:gensys} involving stochastic uncertainties, typically, a sum of cost-per-stage functions over \(n\) steps grows with \(n\) rather quickly, and the expected total cost
		\[
			\EE_x\biggl[\sum_{t=0}^{+\infty} c(x_t, u_t)\biggr]
		\]
		may not be suitable for measuring performance. For instance, in the case of the standard optimal linear quadratic regulator, under the optimal policy the quantity \(\EE_x\bigl[\sum_{t=0}^n (x_t\transp Q x_t + u_t\transp R u_t)\bigr]\) grows linearly with \(n\) (under mild assumptions on the non-negative definite states- and control-weight matrices \(Q\) and \(R\), respectively); consequently, the expected total cost over an infinite horizon is not bounded. A more appropriate measure of performance is the \emph{long-run expected average cost}, measured in terms of the cost-per-stage function \(c\), defined by
		\begin{equation}
		\label{e:expected avg cost}
			\limsup_{n\to\infty} \frac{1}{n+1}\EE_x\biggl[\sum_{t=0}^n c(x_t, u_t)\biggr].
		\end{equation}
		This is the performance index that we adopt here. In particular, the quantity in \eqref{e:expected avg cost} is well-defined for the linear quadratic problem under mild hypothesis (stabilizability of the underlying linear system). The intuition is clear: the quantity in \eqref{e:expected avg cost} measures the cost averaged over time and averaged across all possible realizations of the process \((x_t)_{t\in\Nz}\). Observe that all phenomena that occur over a finite temporal horizon, or are transient in the sense that they asymptotically die out, do not affect the index \eqref{e:expected avg cost}.

		We next provide our second main result---an estimate of performance under the receding horizon policy \(\rhpolicy\) in terms of the long-run expected average cost. It also shows how stability of the closed-loop process influences the long-run expected average cost.
		\begin{theorem}
		\label{t:ergodic cost}
			Consider the controlled system \eqref{e:gensys} with its accompanying data, and the optimal control problem \eqref{e:MPC problem}. Let \(\rhpolicy \Let (\pi_0\opt, \pi_0\opt, \ldots)\) denote the receding horizon policy derived from \eqref{e:MPC problem}, and for a measurable feedback \(\tilde g\) we define
			\begin{equation}
			\label{e:T def}
				T_{\tilde g}(z) \Let c(z, \tilde g(z)) - c_F(z) + \EE\bigl[ c_F\circ f(z, \tilde g(z), w_0)\bigr],\quad z\in\R^d.
			\end{equation}
			Then the following hold:
			\begin{enumerate}[label={\rm (\ref{t:ergodic cost}.\roman*)}, align=left, leftmargin=*, widest=ii]
				\item \label{t:ergodic cost:ineq} If Assumption \ref{a:opt existence} holds, then for every \(x\in\R^d\) and \(k\in\N\),
					\begin{equation}
					\label{e:ergodic cost:ineq}
					\begin{aligned}
						\EE^{\rhpolicy}_x\biggl[\sum_{\ell=0}^{k} c(x_\ell, u_\ell)\biggr] & \le V_N\opt(x) - \EE^{\rhpolicy}_x\bigl[V_N\opt(x_{k+1})\bigr]\\
						& \quad + \sum_{\ell=0}^k \EE^{\rhpolicy}_x\Bigl[ \EE^{\pi\opt}\bigl[ T_{\tilde g}(x_{\ell+N})\,\big|\, x_\ell\bigr]\Bigr].
					\end{aligned}
					\end{equation}
				\item \label{t:ergodic cost:bound} If Assumption \ref{a:feedback} holds and
					\begin{equation}
					\label{e:ergodic cost:asymptote}
						\lim_{n\to+\infty} \frac{\EE^{\rhpolicy}_x\bigl[V_N\opt(x_n)\bigr]}{n} = 0 \quad \text{for all \(x\in\R^d\)},
					\end{equation}
					then
					\begin{equation}
					\label{e:ergodic cost:bound}
						\limsup_{k\to+\infty} \frac{1}{k+1} \EE^{\rhpolicy}_x\biggl[\sum_{\ell=0}^{k} c(x_\ell, u_\ell)\biggr] \le b\quad\text{for all \(x\in\R^d\)}.
					\end{equation}
					In particular, if for all \(x\in\R^d\) the sequence \(\bigl(\EE^{\rhpolicy}_x\bigl[V_N\opt(x_t)\bigr]\bigr)_{t\in\Nz}\) is bounded, then \eqref{e:ergodic cost:bound} holds.
			\end{enumerate}
		\end{theorem}

		\begin{remark}
		\label{r:performance}
		\mbox{}
		\vspace*{-2ex}
			\begin{enumerate}[label=\alph*), align=left, widest=b, leftmargin=*]
				\item The estimate \eqref{e:ergodic cost:ineq} in part \ref{t:ergodic cost:ineq} is quite general, and holds under the mild hypotheses \ref{a:opt existence} of a technical nature. The part \ref{t:ergodic cost:bound} follows from \ref{t:ergodic cost:ineq} under the additional Assumption \ref{a:feedback}, as will be evident from the proof of Theorem \ref{t:ergodic cost} presented in Appendix \ref{s:app:perf}.
				\item We reiterate that the bound on the expected average cost in \eqref{e:ergodic cost:bound} corresponds to the receding horizon policy \(\rhpolicy\), not the stabilizing policy \((g, g, \ldots)\). Indeed, although Assumption \ref{a:feedback} stipulates the existence of a stabilizing feedback \(g\), this feedback controller is never applied. From the proof of Theorem \ref{t:ergodic cost} one sees that the condition \eqref{e:drift condition rh} in Theorem \ref{t:drift} plays a crucial role in \ref{t:ergodic cost:bound}.
				\item The condition \eqref{e:ergodic cost:asymptote} is technical in nature. In most practical cases, one ensures boundedness of \(\bigl(\EE^{\rhpolicy}_x\bigl[V_N\opt(x_t)\bigr]\bigr)_{t\in\Nz}\) with the help of Theorem \ref{t:drift}, and consequently \eqref{e:ergodic cost:bound} holds.
				\item Theorem \ref{t:ergodic cost} provides a quantitative bound on the performance measured in terms of the cost-per-stage function \(c\). Since the stability conditions presented in \secref{s:stab:cost} also involve the function \(c\), it is possible to establish a connection between the performance bound \eqref{e:ergodic cost:bound} with the stability conditions in Assumption \ref{a:feedback}. In general, establishing such a connection may not be possible if stability under the receding horizon policy \(\rhpolicy\) is ensured by means of adjoining an appropriate constraint, as discussed in \ref{stab:constraint} above; indeed, such a constraint may have no relation whatsoever to the cost functions \(c\) and \(c_F\). Nevertheless, the inequality in \eqref{e:ergodic cost:ineq} holds irrespective of whether Assumption \ref{a:feedback} is satisfied or not. Consequently, if a constraint is adjoined to the problem \eqref{e:MPC problem} such that the process \((x_t)_{t\in\Nz}\) under the receding horizon policy \(\rhpolicy\) satisfies \eqref{e:ergodic cost:asymptote}, and the sum on the right-hand side of \eqref{e:ergodic cost:ineq} increases at most linearly with \(k\), then a bound on the long-run expected average cost can be extracted. Naturally, these verifications, which are likely to be case-specific, need additional analysis. 
				\item A direct optimal control problem involving minimization of the long-run expected average cost criterion for dynamical systems, e.g.,
					\begin{equation}
					\label{e:ex avg cost:opt problem}
					\begin{aligned}
						\minimize	& \quad \limsup_{n\to+\infty} \frac{1}{n+1}\EE^\pi_x\biggl[\sum_{t=0}^n c(x_t, u_t)\biggr]\\
						\sbjto		& \quad \begin{cases}
										\pi\in\Pi,\\
										\text{dynamics \eqref{e:gensys}},
							\end{cases}
					\end{aligned}
					\end{equation}
					is generally difficult to solve both analytically and numerically---see \cite[Chapter 5]{ref:hernandez-lerma1} for further details. The numerical value of the bound \(b\) in \eqref{e:ergodic cost:bound} on the performance index \eqref{e:expected avg cost} under the receding horizon technique may be employed as a measure to decide whether to adopt a receding horizon strategy as an alternative to a direct solution to the minimization problem \eqref{e:ex avg cost:opt problem}.
				\RemarkEnd
			\end{enumerate}
		\end{remark}

		\begin{Example}[Example \ref{ex:FL classical} cont'd]
			Consider the linear system in Example \ref{ex:FL classical} and the selection of cost functions \(c\) and \(c_F\) as in the final part of Example \ref{ex:FL classical}. Standard computations, e.g., as in \cite[Chapter 3]{ref:hernandez-lerma1}, show that the function \(V_N\opt\) is quadratic. In conjunction with the final calculations in Example \ref{ex:FL classical}, this shows that \eqref{e:drift condition rh} is a geometric drift condition in this case. Proposition \ref{p:FL classical} implies, therefore, that under receding horizon control derived from \eqref{e:MPC problem}, for each \(x\in\R^d\) the sequence \(\bigl(\EE^{\rhpolicy}_x\bigl[V_N\opt(x_n)\bigr]\bigr)_{n\in\Nz}\) is bounded. Finally, by Theorem \ref{t:ergodic cost} part \ref{t:ergodic cost:bound} we see that for each \(x\in\R^d\)
			\[
				\limsup_{n\to\infty}\frac{1}{n+1} \EE^{\rhpolicy}_x\biggl[\sum_{t=0}^n c(x_t, u_t)\biggr] \le \beta,
			\]
			where \(\beta\) is the constant defined in Example \ref{ex:FL classical}.
			\ExampleEnd
		\end{Example}

		\begin{Example}[Example \ref{ex:linear scalar} cont'd]
			A direct application of Proposition \ref{p:geometric drift} does not appear to be possible. However, one can verify that for each \(x\in\R^d\) the sequence \(\bigl(\EE^{\rhpolicy}_x\bigl[V_N\opt(x_t)\bigr]\bigr)_{t\in\Nz}\) is bounded by following the arguments in \cite{ref:RamChaMilHokLyg-10} with the support of Proposition \ref{p:PR conditions}. This in turn implies by Theorem \ref{t:ergodic cost} that the long-run expected average cost is finite under the receding horizon policy extracted from \eqref{e:scalar opt problem}.
			\ExampleEnd
		\end{Example}

		\begin{Example}[Example \ref{ex:linear scalar re} cont'd]
			In this case it is possible to directly apply Proposition \ref{p:geometric drift}; we skip the standard computations needed to verify its hypotheses. Consequently, the sequence \(\bigl(\EE^{\rhpolicy}_x\bigl[V_N\opt(x_t)\bigr]\bigr)_{t\in\Nz}\) is bounded. By Theorem \ref{t:ergodic cost} this implies that the long-run expected average cost is finite under the receding horizon policy extracted from \eqref{e:scalar opt problem}.
			\ExampleEnd
		\end{Example}

	\appendix

	\section{}
	\label{s:app:stab}
		This appendix collects our proofs of the results presented in \secref{s:stab}.

		\begin{proof}[Proof of Proposition \ref{p:FL classical}]
			A proof of this proposition is standard; we include it merely for completeness. The assertion follows from the Markovian property of \((x_t)_{t\in\Nz}\) and an iteration scheme as follows:
			\begin{align*}
				\EE_x[V(x_t)] & = \EE_x\bigl[\EE[V(x_t)\mid x_{t-1}]\bigr]\\
					& \le \EE_x\bigl[\lambda_\circ V(x_{t-1})\indic{\R^d\setminus K}(x_{t-1}) + \beta \indic{K}(x_{t-1})\bigr]\\
					& \le \lambda_{\circ} \EE_x\bigl[V(x_{t-1})\bigr] + \beta \PP_x(x_{t-1}\in K)\\
					& \vdots\\
					& \le \lambda_\circ^t V(x) + \beta \sum_{k=0}^{t-1} \lambda_\circ^{t-1-k} \PP_x(x_{k}\in K)\\
					& \le \lambda_\circ^t V(x) + \frac{\beta}{1-\lambda_\circ}.\qedhere
			\end{align*}
		\end{proof}
		
		\begin{proof}[Proof of Proposition \ref{p:PR conditions}]
			This is a special case of \cite[Theorem 1]{ref:PemRos-99}.
		\end{proof}

		\begin{proof}[Proof of Theorem \ref{t:drift}]
			From the last \(N-1\) elements of the optimal policy \(\pi\opt\) we derive the \(N\)-length policy \(\tilde\pi \Let (\pi_1\opt, \ldots, \pi_{N-1}\opt, g)\), where \(g:\R^d\lra\ConSet\) is the feedback function in Assumption \ref{a:feedback}. Let \(\tilde V_N(x) \Let V_N(\tilde\pi, x)\). Recall that the sequence \((x_t\opt)_{t=0}^{N-1}\) is generated by the recursion \eqref{e:optimal traj}. Then, by optimality of the policy \(\pi\opt\),
			\begin{align*}
				& \EE_x\bigl[V_N\opt(x_1\opt)\bigr] - V_N\opt(x)\\
				& \le \EE_x\bigl[ \tilde V_N(x_1\opt)\bigr] - V_N\opt(x)\\
				& = \EE_x\biggl[ \EE\biggl[ \sum_{t=1}^{N-1} c(x_t\opt, \pi_t\opt(x_t\opt)) + c(x_N\opt, g(x_N\opt)) + c_F\bigl( f(x_N\opt, g(x_N\opt), w_N)\bigr) \,\bigg|\, x_1\opt\biggr]\biggr]\\
				& \qquad - \EE_x\biggl[ \sum_{t=0}^{N-1} c(x_t\opt, \pi_t\opt(x_t\opt)) + c_F(x_N\opt)\biggr] \\
				& = - c(x, \pi_0\opt(x)) + \EE_x\Bigl[ c(x_N\opt, g(x_N\opt)) + c_F\circ f(x_N\opt, g(x_N\opt), w_N) - c_F(x_N\opt)\Bigr].
			\end{align*}
			If the condition \ref{a:feedback} holds, then the tower property of conditional expectations implies that
			\begin{align*}
				& \EE_x\Bigl[c(x_N\opt, g(x_N\opt)) + c_F\circ f(x_N\opt, g(x_N\opt), w_N) - c_F(x_N\opt)\Bigr]\\
				& = \EE_x\Bigl[\EE\bigl[c(x_N\opt, g(x_N\opt)) + c_F\circ f(x_N\opt, g(x_N\opt), w_N) - c_F(x_N\opt)\,\big|\,\{x_\ell\opt\}_{\ell=1}^N\bigr]\Bigr]\\
				& \le \EE_x\Bigl[ b \indic{\{x_N\opt\in K\}}\Bigr]\\
				& = b \PP_x(x_N\opt\in K).
			\end{align*}
			Substituting back we obtain
			\[
				\EE_x\bigl[V_N\opt(x_1\opt)\bigr] - V_N\opt(x) \le -c(x, \pi_0\opt(x)) + b,
			\]
			which proves \eqref{e:drift condition}. The assertion \eqref{e:drift condition rh} follows from the facts that \(\rhpolicy = (\pi_0\opt, \pi_0\opt, \ldots)\), and that the closed-loop process is Markovian under \(\rhpolicy\).
		\end{proof}

		\begin{proof}[Proof of Proposition \ref{p:V vs final-cost}]
			Fix \(x\in\R^d\). The first inequality \(c(x, \pi_0\opt(x)) \le V_N\opt(x)\) follows immediately from the definition of \(V_N\opt\). We now prove the second inequality. Let the sequence \((x^g_t)_{t=1}^N\) be defined by 
			\[
				x^g_{t} = \begin{cases}
					x	& \text{if \(t = 0\)},\\
					f\bigl(x^g_{t-1}, g(x^g_{t-1}), w_{t-1}\bigr) & \text{if \(t=1, \ldots, N\)}.
				\end{cases}
			\]
			In view of \ref{e:stabilizing condition}, we have
			\begin{align*}
				c_F(x) - V_N\opt(x) & = c_F(x) - c(x, g(x)) - V_N\opt(x) + c(x, g(x))\\
					& \ge \EE\bigl[c_F\circ f(x, g(x), w_0)\bigr] - V_N\opt(x) + c(x, g(x)) - b\\
					& = \EE\bigl[c_F(x^g_1) - c(x^g_1, g(x^g_1))\,\big|\, x^g_0 = x\bigr] - V_N\opt(x)\\
					& \qquad + \EE\bigl[c(x^g_1, g(x^g_1))\,\big|\, x^g_0 = x\bigr] + c(x^g_0, g(x^g_0)) - b\\
					& \ge \EE\bigl[c_F(x^g_2) \,\big|\, x^g_0 = x\bigr] - V_N\opt(x)\\
					& \qquad + \sum_{\ell=0}^1 \EE\bigl[c(x^g_\ell, g(x^g_\ell))\,\big|\, x^g_0 = x\bigr] - 2b\\
					& \vdots\\
					& \ge \EE\bigl[c_F(x^g_N)\,\big|\, x^g_0 = x\bigr] - V_N\opt(x) - Nb\\
					& \qquad + \sum_{\ell=0}^{N-1} \EE\bigl[c(x^g_\ell, g(x^g_\ell))\,\big|\, x^g_0 = x\bigr].
			\end{align*}
			Since \(\overset{N\text{ times}}{\overbrace{(g, \ldots, g)}}\) is a sub-optimal policy for the optimal control problem \eqref{e:MPC problem},
			\[
				V_N\opt(x) \le \EE\biggl[\sum_{\ell=0}^{N-1} c(x^g_\ell, g(x^g_\ell)) + c_F(x^g_N) \,\bigg|\, x^g_0 = x\biggr],
			\]
			and therefore, \(c_F(x) - V_N\opt(x) \ge -Nb\). Since \(x\) is arbitrary, the assertion follows.
		\end{proof}

		\begin{proof}[Proof of Proposition \ref{p:geometric drift}]
			Fix \(x\in\R^d\). We know from Theorem \ref{t:drift} that
			\[
				\EE^{\rhpolicy}_x\bigl[V_N\opt(x_1)\bigr] - V_N\opt(x) \le -c(x, \pi_0\opt(x)) + b.
			\]
			By hypothesis, \(c(x, \pi_0\opt(x)) = c_s(x) + c_c(\pi_0\opt(x)) \ge c_s(x)\), and if \(x\not\in K\), then \(c_s(x) \ge \alpha c_F(x)\). Thus,
			\[
				\EE^{\rhpolicy}_x\bigl[V_N\opt(x_1)\bigr] - V_N\opt(x) \le -c_s(x) + b \le -\alpha c_F(x) + b\quad\text{if }x\not\in K.
			\]
			In view of \eqref{e:V vs final cost}, we see that if \(x\not\in K\), then \(-\alpha c_F(x) + b \le -\alpha V_N\opt(x) + b(1 + \alpha N)\). In other words,
			\[
				\EE^{\rhpolicy}_x\bigl[V_N\opt(x_1)\bigr] - V_N\opt(x) \le -\alpha V_N\opt(x) + b (1 + \alpha N)\quad\text{for all \(x\not\in K\)}.
			\]
			Since \(\lim_{\norm{z}\to+\infty} c_s(z) = +\infty\), our hypotheses show that \(\lim_{\norm{z}\to+\infty} c(z, \pi_0\opt(z)) = +\infty\), and from \eqref{e:V vs final cost} it follows that \(\lim_{\norm{z}\to+\infty} V_N\opt(z) = +\infty\). By definition of a limit, therefore, there must exist an closed ball \(K'\) around \(0\in\R^d\) of radius large enough, such that \(V_N\opt(z) \ge 2(\alpha\inverse + N)\) for all \(z\not\in K'\). Substituting back we see that
			\[
				\EE^{\rhpolicy}_x\bigl[V_N\opt(x_1)\bigr] - V_N\opt(x) \le -\frac{\alpha}{2} V_N\opt(x)\quad\text{for all }x\not\in K',
			\]
			which is a geometric drift condition outside the compact set \(K'\). The particular case follows from Proposition \ref{p:FL classical}.
		\end{proof}

		\begin{proof}[Proof of Proposition \ref{p:ortho stability}]
			Fix \(t\in\Nz\). The state recursion \eqref{e:linsys} shows that
			\begin{align*}
				x_{\reachindex(t+1)} & = A^\reachindex x_{\reachindex t} + \reachab(A, B) \pmat{u_{\reachindex t}\\ \vdots\\ u_{\reachindex(t+1)-1}} + \reachab(A, I_d) \pmat{w_{\reachindex t}\\ \vdots \\ w_{\reachindex(t+1)-1}}\\
				& \teL A^\reachindex x_{\reachindex t} + \reachab(A, B) \ol{u_{\reachindex t}} + \reachab(A, I_d) \ol{w_{\reachindex t}}.
			\end{align*}
			Let \(\tilde\lambda \Let \rho - U_{\max}\) and let \(\lambda_\circ \Let \epower{\tilde\lambda}\); by hypothesis, \(\tilde\lambda < 0\) and \(\lambda_\circ \in\:]0, 1[\). On the event \(\{\norm{x_{\reachindex t}} > 2R\}\), we have
			\begin{align*}
				& \frac{\EE\bigl[\epower{\norm{x_{\reachindex (t+1)}}}\,\big|\, x_{\reachindex t}\bigr]}{\epower{\norm{x_{\reachindex t}}}}\\
				& \quad = \EE\Bigl[ \exp\bigl( \norm{A^\reachindex x_{\reachindex t} + \reachab(A, B) \ol{u_{\reachindex t}} + \reachab(A, I_d) \ol{w_{\reachindex t}}} - \norm{x_{\reachindex t}}\bigr) \,\Big|\, x_{\reachindex t}\Bigr]\\
				& \quad = \epower{\norm{A^\reachindex x_{\reachindex t} - \sat_{U_{\max}}(A^\reachindex x_{\reachindex t})} - \norm{A^\reachindex x_{\reachindex t}}}	\EE\bigl[\epower{\norm{\reachab(A, I_d) \ol{w_{\reachindex t}}}} \,\big|\, x_{\reachindex t}\bigr]\quad\text{since \(A\) is orthogonal}\\
				& \quad = \epower{-U_{\max}} \EE\!\left[\exp\!\left(\norm{\reachab(A, I_d) \pmat{w_0\\ \vdots \\ w_{\reachindex -1}}}\right)\right] \quad\text{since \((w_s)_{s\in\Nz}\) is i.i.d.}\\
				& \quad = \epower{-U_{\max} + \rho}\\
				& \quad = \lambda_\circ < 1 \quad \text{by hypothesis}.
			\end{align*}
			Since \(t\) is arbitrary, the first claim follows. By hypothesis, the control actions \(\pmat{u_{\reachindex t}\transp & \cdots & u_{\reachindex (t+1)-1}\transp}\transp\) depend only on \(x_{\reachindex t}\) for each \(t\in\Nz\). Thus, the process \((x_{\reachindex t})_{t\in\Nz}\) is Markovian under the control actions in \eqref{e:ortho stability:controls}, as can be seen by directly verifying \eqref{e:Markovian proof}. By Proposition \ref{p:FL classical} we see that for each \(x\in\R^d\) the sequence \(\bigl(\EE_x\bigl[\epower{\norm{x_{\reachindex t}}}\bigr]\bigr)_{t\in\Nz}\) is bounded. It remains to move from the \(\reachindex\)-subsampled process \((x_{\reachindex t})_{t\in\Nz}\) to the original process \((x_t)_{t\in\Nz}\). But standard arguments, along the lines of the proof of the main theorem in \cite{ref:RamChaMilHokLyg-10}, employing the triangle inequality and monotonicity of the function \(\epower{(\cdot)}\) shows that the sequence \(\bigl(\EE_x\bigl[\epower{\norm{x_{t}}}\bigr]\bigr)_{t\in\Nz}\) is bounded.

			For \(x\in\R^d\) let \(C = C(x) > 0\) be such that \(\sup_{t\in\Nz} \EE_x\bigl[\epower{\norm{x_t}}\bigr] \le C\). Then for all \(t\in\Nz\),
			\begin{align*}
				C	& \ge \EE_x\bigl[\epower{\norm{x_t}}\bigr] = \EE_x\biggl[ \int_0^{\norm{x_t}} \epower{r}\,\drv r\biggr] + 1\\
					& = \EE_x\biggl[\int_0^{+\infty} \indic{\{\norm{x_t} > r\}} \epower{r}\,\drv r\biggr] + 1\\
					& = \int_0^{+\infty} \PP_x\bigl(\norm{x_t} > r\bigr) \epower{r} \,\drv r + 1\quad \text{by Fubini's theorem}.
			\end{align*}
			This shows that \(\PP_x(\norm{x_t} > r)\) must decay, for large values of \(r\), faster than \(\epower{-r}\) uniformly for all \(t\in\Nz\), and the assertion follows.
		\end{proof}

	\section{}
	\label{s:app:perf}
		This appendix contains our proof of Theorem \ref{t:ergodic cost}. For two policies \(\pi_{0:k_1-1}\) and \(\pi'_{0:k_2-1}\) of length \(k_1\) and \(k_2\), respectively, we define their concatenation 
		\[
			\pi_{0:k_1-1}\join\pi'_{0:k_2-1} \Let (\pi_0, \ldots, \pi_{k_1-1}, \pi'_0, \ldots, \pi'_{k_2-1}).
		\]

		\begin{proof}[Proof of Theorem \ref{t:ergodic cost}]
			\ref{t:ergodic cost:ineq}: We adapt certain ideas from \cite{ref:recstrat} to the context of long-run expected average cost. Suppose that Assumption \ref{a:opt existence} holds, and fix \(n\in\N\). Conditional on \(x_{n+1} = x'\in\R^d\), by definition of optimality,
			\begin{align*}
				& \EE^{\pi\opt_{1:N-1}\join \tilde g} \biggl[ \sum_{\ell=n+1}^{n+N} c(x_\ell, u_\ell) + c_F(x_{n+N+1})\,\bigg|\, x_{n+1} = x' \biggr]\\
				& \quad \ge \EE^{\pi_{0:N-1}\opt}\biggl[ \sum_{\ell=n+1}^{n+N} c(x_\ell, u_\ell) + c_F(x_{n+N+1})\,\bigg|\, x_{n+1} = x' \biggr]
			\end{align*}
			Conditional on \(x_n = y\), therefore,
			\begin{align*}
				& \EE^{\pi_0\opt\join\pi_{1:N-1}\opt\join \tilde g}\biggl[ \biggl( \sum_{\ell=n}^{n+N-1} c(x_\ell, u_\ell) + c_F(x_{n+N}) \biggr) - c(x_n, u_n) \,\bigg|\, x_n = y \biggr]\\
				& \quad + \EE^{\pi_0\opt\join\pi_{1:N-1}\opt\join \tilde g}\Bigl[ c_F(x_{n+N+1}) - c_F(x_{n+N}) + c(x_{n+N}, u_{n+N}) \,\Big|\, x_n = y\Bigr]\\
				& \qquad \ge \EE^{\pi_0\opt\join\pi_{0:N-1}}\biggl[ \sum_{\ell=n+1}^{n+N} c(x_\ell, u_\ell) + c_F(x_{n+N+1})\,\bigg|\, x_n = y\biggr],
			\end{align*}
			which implies
			\begin{align*}
				& \EE^{\pi_{0:N-1}\opt\join \tilde g}\biggl[ \sum_{\ell=n}^{n+N-1} c(x_i, u_i) + c_F(x_{n+N})\,\bigg|\, x_n = y\biggr]\\
				& \quad + \EE^{\pi_{0:N-1}\opt\join \tilde g}\Bigl[ c_F(x_{n+N+1}) - c_F(x_{n+N}) + c(x_{n+N}, u_{n+N})\,\Big|\, x_n = y\Bigr]\\
				& \qquad \ge c(y, \pi_0\opt(y)) + \EE^{\pi_0\opt\join\pi_{0:N-1}\opt}\biggl[\sum_{\ell=n+1}^{n+N} c(x_i, u_i) + c_F(x_{n+N+1}))\,\bigg|\, x_n = y\biggr].
			\end{align*}
			Rearranging terms, conditional on \(x_n = y\), we get
			\begin{align*}
				c(y, \pi_0\opt(y)) & \le V_N\opt(y) - \EE^{\pi_0\opt}\bigl[ V_N\opt(x_{n+1})\,\big|\, x_n = y\bigr]\\
					& \quad + \EE^{\pi\opt}\Bigl[ c(x_{n+N}, \tilde g(x_{n+N})) - c_F(x_{n+N})\\
					& \qquad \qquad + \EE\bigl[ c_F\circ f(x_{n+N}, \tilde g(x_{n+N}), w_{n+N})\,\big|\, x_{n+N}\bigr]\,\Big|\, x_n = y\Bigr].
			\end{align*}
			Suppose now that the receding horizon policy \(\rhpolicy\) is applied. Since the closed-loop process \((x_t)_{t\in\Nz}\) under \(\rhpolicy\) is Markovian, taking expectations under the policy \(\rhpolicy\), we get
			\begin{align*}
				\EE^{\rhpolicy}_x\bigl[ c(x_n, \pi_0\opt(x_n))\bigr]	& \le \EE^{\rhpolicy}_x\bigl[V_N\opt(x_n)\bigr] - \EE^{\rhpolicy}_x\bigl[V_N\opt(x_{n+1})\bigr] + \EE^{\rhpolicy}_x\bigl[ \EE^{\pi\opt}\bigl[T_{\tilde g}(x_{n+N})\,\big|\, x_n\bigr]\bigr],
			\end{align*}
			whence, summing from \(n = 0\) through \(n = k\) we arrive at
			\begin{align*}
				\EE^{\rhpolicy}_x\biggl[ \sum_{n=0}^k c(x_n, \pi_0\opt(x_n))\biggr]	& \le V_N\opt(x) - \EE^{\rhpolicy}_x\bigl[V_N\opt(x_{k+1})\bigr] + \sum_{n=0}^k \EE^{\rhpolicy}_x\bigl[ \EE^{\pi\opt}\bigl[ T_{\tilde g}(x_{n+N})\,\big|\, x_n\bigr]\bigr],
			\end{align*}
			as asserted.

			\ref{t:ergodic cost:bound} Suppose that Assumption \ref{a:feedback} holds, and let \(\tilde g = g\). Fix \(x\in\R^d\). From \eqref{e:ergodic cost:ineq}, it follows that
			\begin{align*}
				\frac{1}{k+1}	& \EE^{\rhpolicy}_x\biggl[ \sum_{n=0}^k c(x_n, \pi_0\opt(x_n))\biggr]\\
					& \le \frac{1}{k+1}\Bigl( V_N\opt(x) - \EE^{\rhpolicy}_x\bigl[V_N\opt(x_{k+1})\bigr]\Bigr) + \frac{1}{k+1}\sum_{n=0}^k \EE^{\rhpolicy}_x\bigl[ \EE^{\pi\opt}\bigl[ T_g(x_{n+N})\,\big|\, x_n\bigr]\bigr]\\
					& \le \frac{1}{k+1}\Bigl(V_N\opt(x) - \EE^{\rhpolicy}_x\bigl[V_N\opt(x_{k+1})\bigr]\Bigr) + b\quad \text{by Assumption \ref{a:feedback}},
			\end{align*}
			which implies that
			\begin{align*}
				\limsup_{k\to+\infty} \frac{1}{k+1}\EE^{\rhpolicy}_x\biggl[ \sum_{n=0}^k c(x_n, \pi_0\opt(x_n))\biggr] & \le \limsup_{k\to+\infty} \frac{V_N\opt(x) - \EE^{\rhpolicy}_x\bigl[V_N\opt(x_{k+1})\bigr]}{k+1} + b\\
					& = b\quad\text{by hypothesis}.
			\end{align*}
			The final claim follows at once from the preceding.
		\end{proof}
\bigskip

\def\cprime{$'$}

\bigskip
\bigskip

\end{document}